\documentclass{article}

\usepackage{PRIMEarxiv}

\usepackage[utf8]{inputenc} % allow utf-8 input
\usepackage[T1]{fontenc}    % use 8-bit T1 fonts
\usepackage{hyperref}       % hyperlinks
\usepackage{url}            % simple URL typesetting
\usepackage{booktabs}       % professional-quality tables
\usepackage{}
\usepackage{amsfonts}       % blackboard math symbols
\usepackage{nicefrac}       % compact symbols for 1/2, etc.
\usepackage{microtype}      % microtypography
\usepackage{amsmath}
\usepackage{amssymb}
\usepackage{fancyhdr}       % header
\usepackage{graphicx}       % graphics
\graphicspath{{media/}}     % organize your images and other figures under media/ folder
\usepackage{subfig}
\usepackage{tikz}
\usepackage{amsthm}
\usepackage{subcaption}
\usepackage{caption}
\usetikzlibrary{arrows}
\theoremstyle{definition}
\newtheorem{theorem}{Theorem}
\newtheorem{lemma}{Lemma}
\newtheorem{cor}{Corollary}

\usepackage{color}
\usepackage{dutchcal}
\usepackage{textcomp} 
\usepackage{float}
\usetikzlibrary{tikzmark, positioning}
\usepackage{mathabx}

\def\CircleArrowleft{\ensuremath{%
  \reflectbox{\rotatebox[origin=c]{180}{$\circlearrowleft$}}}}
\def\CircleArrowright{\ensuremath{%
  \reflectbox{\rotatebox[origin=c]{180}{$\circlearrowright$}}}}

\title{Hexagonal Picture Scanning Automata
}

\author{
 Deepalakshmi D \\
   Amal Jyothi College of Engineering,  Kanjirappally, Kerala,
  India\\
  APJ Abdul Kalam Technological University,Thiruvananthapuram,
  India\\
  \texttt{ddeepalakshmi@amaljyothi.ac.in} \\
   \And
  Lisa Mathew$^*$ \\
   Amal Jyothi College of Engineering,  Kanjirappally, Kerala,
  India\\
    \texttt{lisamathew@amaljyothi.ac.in} \\
}

\begin{document}
\maketitle

\begin{abstract}
Two new classes of finite automata, called General hexagonal Boustrophedon finite automata and General hexagonal returning finite automata operating on hexagonal grids, are introduced and analyzed. The work establishes the theoretical foundations for these automata models, examines their computational properties, and investigates the relationships and equivalences between the language families they define. The research contributes to the broader understanding of two-dimensional automata theory by extending classical finite automaton concepts to hexagonal geometric structures with specialized traversal patterns.
\end{abstract}

% keywords can be removed

\section{Introduction}
Two-dimensional automata theory extends classical string processing to higher-dimensional structures, providing fundamental insights into pattern recognition and image processing \cite{inoue1991survey,fernau2015scanning}. Early work by Siromoney et al. \cite{siromoney1973picture,krithivasan1974array,subramanian1989siromoney} established the theoretical framework for systematic study of array and picture processing.

Traditional two-dimensional language recognition focuses on rectangular grids with standard traversal patterns. However, alternative geometric arrangements and traversal strategies can lead to fundamentally different classes of recognizable languages \cite{giammarresi1997two}. The work by Fernau et al. \cite{fernau2015scanning,picturescanning} on boustrophedon scanning methods demonstrates that alternating-direction traversal patterns offer computational advantages in two-dimensional data processing.

Hexagonal grid structures provide a compelling alternative to rectangular arrangements. Unlike rectangular grids with four-neighbor connectivity, hexagonal grids offer six-neighbor connectivity, which models natural phenomena more closely and provides computational advantages in certain applications. Siromoney et al. \cite{krithivasan1974array,siromoney1976hexagonal} established the foundational framework for hexagonal arrays, while specific hexagonal automaton models have been explored extensively \cite{anitha2017hexagonal,pruvsa2004two}. The work on Boustrophedon automaton scanning rectangular pictures motivated a corresponding study on hexagonal pictures.

Studies on Picture-walking automata \cite{kari2011survey}  show that movement patterns significantly impact computational power and recognizable language classes. This motivates investigating how boustrophedon and returning movement strategies perform on hexagonal grid structures.

The rest of this paper is organized as follows. Section 2 discusses preliminary notions related to the results. In Sections 3 and 4 we introduce and analyze General Hexagonal Boustrophedon Finite Automata and General Hexagonal Returning Finite Automata. Section 5 examines equivalences between these language families. Section 6 presents conclusions and future work directions.

\section{Preliminaries}
\label{prelims}
\subsection{Hexagonal Picture languages}
Given a finite alphabet $\Sigma$, a hexagonal picture $P$ is defined as a hexagonal arrangement of symbols drawn from $\Sigma$. Figure \ref{hex} depicts a hexagonal picture over the alphabet $\Sigma =\{a\}$ 
 along with a bordered hexagonal picture over the same alphabet (surrounded by $\# \notin\Sigma$).  We denote the collection of all possible hexagonal pictures over alphabet $\Sigma$ as $\Sigma^{**H}$. A hexagonal picture language $L$ is any subset of $\Sigma^{**H}$. 
$\Sigma^{++H}$ represents the set of all non-empty hexagonal pictures over $\Sigma$ i.e. $\Sigma^{++H} = \Sigma^{**H} - \{\lambda\}$, where $\lambda$ denotes the empty picture.
\begin{figure}[h]
     \centering
  \begin{minipage}{0.45\textwidth}
            \centering
\begin{tikzpicture}[rotate=90, scale=0.6]
% Smaller hexagon: max row width is 3
\def\n{2}
% Upper part of hexagon
\foreach \i in {0,...,2} {
    \pgfmathsetmacro{\num}{\n + \i}
    \pgfmathsetmacro{\startx}{-\i*0.5}
    \foreach \j in {0,...,\num} {
        \node at ({\startx + \j}, {-0.8*\i}) {\textit{a}};
    }
}
% Lower part of hexagon
\foreach \i in {1,...,2} {
    \pgfmathsetmacro{\num}{\n + 2 - \i}
    \pgfmathsetmacro{\startx}{- (2 - \i)*0.5}
    \foreach \j in {0,...,\num} {
        \node at ({\startx + \j}, {-0.8*(\i + 2)}) {\textit{a}};
    }
}
\end{tikzpicture}
            \end{minipage}
     \hfill
     \begin{minipage}{0.45\textwidth} 
          \centering
\begin{tikzpicture}[rotate=90, scale=0.6]
\def\n{3} % radius
% Upper part (including middle row)
\foreach \i in {0,...,3} {
    \pgfmathsetmacro{\num}{\n + \i}
    \pgfmathsetmacro{\startx}{-\i*0.5}
    \foreach \j in {0,...,\num} {
        \ifnum\i=0
            \node at ({\startx + \j}, {-0.8*\i}) {\#};
        \else
            \ifnum\j=0
                \node at ({\startx + \j}, {-0.8*\i}) {\#};
            \else
                \ifnum\j=\num
                    \node at ({\startx + \j}, {-0.8*\i}) {\#};
                \else
                    \node at ({\startx + \j}, {-0.8*\i}) {\textit{a}};
                \fi
            \fi
        \fi
    }
}
% Lower part
\foreach \i in {1,...,3} {
    \pgfmathsetmacro{\num}{\n + 3 - \i}
    \pgfmathsetmacro{\startx}{- (3 - \i)*0.5}
    \foreach \j in {0,...,\num} {
        \ifnum\i=3
            \node at ({\startx + \j}, {-0.8*(\i + 3)}) {\#};
        \else
            \ifnum\j=0
                \node at ({\startx + \j}, {-0.8*(\i + 3)}) {\#};
            \else
                \ifnum\j=\num
                    \node at ({\startx + \j}, {-0.8*(\i + 3)}) {\#};
                \else
                    \node at ({\startx + \j}, {-0.8*(\i + 3)}) {\textit{a}};
                \fi
            \fi
        \fi
    }
}
\end{tikzpicture}
      
     \end{minipage}
     \caption{a)Hexagonal picture b) Hexagonal picture surrounded by \#}
                \label{hex}

\end{figure}
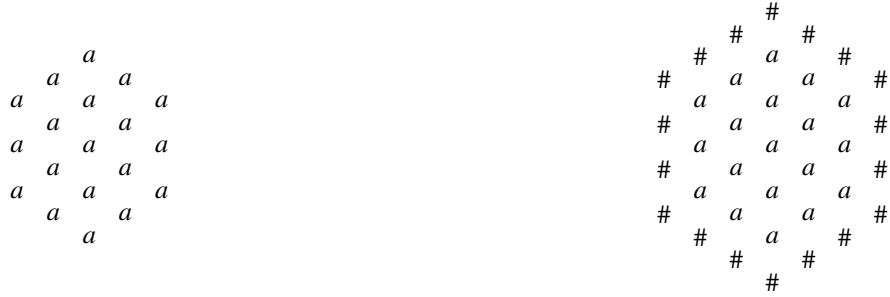
 A hexagonal picture is said to be of size $(l, m, n)$ if it sides are of length $l$, $m$, and $n$ (we assume that opposites sides are of same length). For a picture $p$ of size $(l,m,n)$  the bordered picture $\widehat{p}$ is of size $(l+1,m+1,n+1)$ \cite{anitha2017hexagonal}.

\section{Operations on Hexagonal Pictures}
\label{operations}
 Considering hexagonal pictures as geometrical objects, the following  unary operations have been defined   : $R_n$ - clockwise rotation by $\frac{n\pi}{3}, n=0,1,2\cdots5$,$r_n$- reflection in the line with slope  $\frac{n\pi}{6}$, $n$ = $0,1,2\cdots5$ passing through the center.\\ 
 Note that $R_0 = I$, the identity mapping. Clearly, these 12 operations form a group as illustrated in Table  \ref{comp}.\\
 \begin{table}[h]
     \centering
 \begin{tabular}{|c|c|c|c|c|c|c|c|c|c|c|c|c|}
 \hline
     $\circ$&$I$ &$R_1$&$R_2$&$R_3$&$R_4$&$R_5$&$r_0$&$r_1$&$r_2$&$r_3$&$r_4$&$r_5$\\
     \hline
    $I$&$I$&$R_1$&$R_2$ & $R_3$ & $R_4$&$R_5$&$r_0$&$r_1$&$r_2$&$r_3$&$r_4$&$r_5$\\
    \hline
$R_1$&$R_1$&$R_2$&$R_3$&$R_4$&$R_5$&$I$&$r_1$&$r_2$&$r_3$&$r_4$&$r_5$&$r_0$\\ 
   \hline
   $R_2$&$R_2$&$R_3$&$R_4$&$R_5$&$I$&$R_1$&$r_2$&$r_3$&$r_4$&$r_5$&$r_0$&$r_1$\\
  \hline
   $R_3$&$R_3$&$R_4$&$R_5$&$I$&$R_1$&$R_2$&$r_3$&$r_4$&$r_5$&$r_0$&$r_1$&$r_2$\\
  \hline
  $R_4$&$R_4$&$R_5$&$I$&$R_1$&$R_2$&$R_3$&$r_4$&$r_5$&$r_0$&$r_1$&$r_2$&$r_3$\\
  \hline
  $R_5$&$R_5$&$I$&$R_1$&$R_2$&$R_3$&$R_4$&$r_5$&$r_0$&$r_1$&$r_2$&$r_3$&$r_4$\\
  \hline
 $r_0$&$r_0$&$r_5$&$r_4$&$r_3$&$r_2$&$r_1$&$I$&$R_5$&$R_4$&$R_3$&$R_2$&$R_1$\\
 \hline
 $r_1$&$r_1$&$r_0$&$r_5$&$r_4$&$r_3$&$r_2$&$R_1$&$I$&$R_5$&$R_4$&$R_3$&$R_2$\\
 \hline
 $r_2$&$r_2$&$r_1$&$r_0$&$r_5$&$r_4$&$r_3$&$R_2$&$R_1$&$I$&$R_5$&$R_4$&$R_3$\\
 \hline
 $r_3$&$r_3$&$r_2$&$r_1$&$r_0$&$r_5$&$r_4$&$R_3$&$R_2$&$R_1$&$I$&$R_5$&$R_4$\\
 \hline
 $r_4$&$r_4$&$r_3$&$r_2$&$r_1$&$r_0$&$r_5$&$R_4$&$R_3$&$R_2$&$R_1$&$I$&$R_5$\\
 \hline
$r_5$&$r_5$&$r_4$&$r_3$&$r_2$&$r_1$&$r_0$&$R_5$&$R_4$&$R_3$&$R_2$&$R_1$&$I$ \\
 \hline
\end{tabular}
  \caption{Composition Table for unary operation}
     \label{comp}
 \end{table}
 \begin{table}[h]
     \centering
     \begin{tabular}{|c|c|}
     \hline
    $\circ$&Normal form\\
    \hline
    $I$&$r_1\circ r_1$\\
    \hline
    $R_2$&$R_1\circ R_1$\\
    \hline
    $R_3$&$R_1\circ(R_1\circ R_1)$\\
    \hline
     $R_4$&$(R_1\circ R_1)\circ(R_1\circ R_1)$\\
     \hline
    $R_5$&$R_1\circ((R_1\circ R_1)\circ(R_1\circ R_1))$\\ 
    \hline
    $r_0$&$r_1\circ R_1$\\
    \hline
    $r_2$&$r_1\circ (R_1\circ((R_1\circ R_1)\circ(R_1\circ R_1)))$\\
    \hline
    $r_3$&$r_1\circ ((R_1\circ R_1)\circ(R_1\circ R_1))$\\
    \hline
    $r_4$&$r_1\circ  (R_1\circ(R_1\circ R_1))$\\
    \hline
    $r_5$&$r_1\circ (R_1\circ R_1)$\\
    \hline
        \end{tabular}
     \caption{Normal form for the operators $R_1$ and $r_1$}
     \label{norm}
 \end{table}
 These operations can be partitioned into a set of 6 rotations and a set of 6 reflections which can be expressed in terms of the two primary operations $r_1$ and $R_1$ as given in Table \ref{norm}. These operations can also be applied (picture-wise) to picture languages and (language-wise) to families of picture languages. 
 \section{General Hexagonal Boustrophedon Finite Automata}\label{GHBFA}
 A General Hexagonal Boustrophedon Finite Automata (GHBFA) is an 8 - tuple $M = (Q,\Sigma,R,s,F,\#,\square,D)$, where Q is a finite set of states and is divided into $Q_f$ and $Q_b$, $\Sigma$ is an input alphabet, $R\subseteq Q\times(\Sigma\cup\{\#\})\times Q$ is a finite set of rules of the form  $(q,a,p)$ (usually written as $qa\rightarrow p$), $s\in Q_f$ is the initial state, $F$ is the set of all final states, $\# \notin \Sigma$ denotes the boundary of the picture, $\square$ denotes an erased position, $ D $ denotes the set of all possible directions of movement of the tape head. We add some additional restrictions. If $q\in Q_f$ and $a\in\Sigma$ then $qa\rightarrow p \in R $ is only possible if $p\in Q_f$. Such transitions are called forward transitions and collected in $R_f$. Similarly, if $q\in Q_b$ and $a \in \Sigma$  $qa\rightarrow p\in R$ is possible only when $p\in Q_b $(backward  transitions collected in $R_b$) Finally border transitions (collected in $R_\#$) are of the form $q\#\rightarrow p$ with $q\in Q_f$ if and only if $p\in Q_b$. The set of all possible directions is given by\\

\noindent\begin{align*}
 D=&\Big\{\begin{pmatrix}
s \rotatebox[origin=c]{-60}{$\rightarrow$} & \CircleArrowleft\\
\CircleArrowright& \rotatebox[origin=c]{60}{$\uparrow$}
\end{pmatrix},\begin{pmatrix}
s\rotatebox[origin=c]{60}{$\leftarrow$}&\CircleArrowright \\
\CircleArrowleft & \rotatebox[origin=c]{300}{$\uparrow$}
\end{pmatrix},\begin{pmatrix}
\CircleArrowright &s\downarrow\\
\uparrow&\CircleArrowleft
\end{pmatrix},\begin{pmatrix}
\CircleArrowleft&s\rotatebox[origin=c]{60}{$\uparrow$}\\
\rotatebox[origin=c]{-60}{$\rightarrow$}&\CircleArrowright 
\end{pmatrix},\begin{pmatrix}
 \downarrow&\CircleArrowleft\\
\CircleArrowright&s\uparrow
\end{pmatrix},\begin{pmatrix}
 \rotatebox[origin=c]{300}{$\uparrow$}&\CircleArrowright\\
\CircleArrowleft&s\rotatebox[origin=c]{-60}{$\downarrow$}
\end{pmatrix},\\
&\begin{pmatrix}
 \rotatebox[origin=c]{60}{$\downarrow$}& \CircleArrowright\\
\CircleArrowleft&s\rotatebox[origin=c]{60}{$\uparrow$}
\end{pmatrix},
\begin{pmatrix}
\CircleArrowright &\rotatebox[origin=c]{-60}{$\downarrow$}\\
s\rotatebox[origin=c]{300}{$\uparrow$}&\CircleArrowleft
\end{pmatrix},\begin{pmatrix}
\CircleArrowleft &\downarrow\\
s\uparrow&\CircleArrowright
\end{pmatrix},\begin{pmatrix}
\CircleArrowright &\rotatebox[origin=c]{60}{$\uparrow$}\\
s\rotatebox[origin=c]{60}{$\downarrow$}&\CircleArrowleft
\end{pmatrix},\begin{pmatrix}
 s\rotatebox[origin=c]{300}{$\uparrow$}&\CircleArrowleft\\
\CircleArrowright&\rotatebox[origin=c]{300}{$\downarrow$}
\end{pmatrix},\begin{pmatrix}
 s\downarrow&\CircleArrowright\\
\CircleArrowleft&\uparrow
\end{pmatrix}
\Big\}\end{align*}
 
\noindent where $s$ denotes the start position at one of the corners of the hexagon,  \CircleArrowleft~ indicates movement in  clockwise direction and \CircleArrowright~ indicates movement in  anticlockwise direction. Each $d \in D$ is represented by a $2\times 2$ matrix. The arrow associated with $s$ indicates the initial direction for reading the picture from the starting symbol $s$. When a side has been completely traversed i.e it encounters a $\#$,  it chooses the next direction according  to the next position in the matrix read clockwise.
\par For instance, $\begin{pmatrix}
s \rotatebox[origin=c]{-60}{$\rightarrow$} & \CircleArrowleft\\
\CircleArrowright& \rotatebox[origin=c]{60}{$\uparrow$}
\end{pmatrix}$  indicates that the starting symbol $s$ is positioned at the upper corner of the hexagonal picture. s\rotatebox[origin=c]{-60}{$\rightarrow$}, indicates that we begin reading the picture in the direction \rotatebox[origin=c]{-60}{$\rightarrow$} from $s$. When this row has been traversed completely we move to the next row in clockwise direction and traverse in the opposite direction (arrow), then move to the next row in anticlockwise direction and repeat the entire procedure until the picture has been fully traversed. 
  As the automaton progresses through the entire picture, each symbol is replaced with $\square$ in the order determined by the scanning direction. If the automaton reaches a final state upon completion of the traversal the picture is accepted.
  If $(p,A,f)$ and $(q,A',f)$ are two configurations such that $A$ and $A'$ are identical except for the position $(i,j,k)$, $1\leq i\leq l+1$, $1\leq j\leq m+1$, $1\leq k\leq n+1$, where $A[i, j, k]\in \Sigma$ while $A'[i, j, k]=\square$, we write $(p,A,f)$ $~\vdash_M$ $(q,A',f)$ if $p A[i, j, k]\rightarrow q\in R_f$. \\

Similarly, if $(p,A,f)$ and $(q,A',f)$ are two configurations such that $A$ and $A'$are identical but for one position $(i, j, k)$, $1\leq i\leq l+1$, $1\leq j\leq m+1$, $1\leq k\leq n+1$, where $A[i, j, k]\in \Sigma$ while $A'[i, j, k]=\square$  we write $(p,A,b)$ $~\vdash_M$ $(q,A',b)$ if $pA[i, j, k]\rightarrow q\in R_b$. \\

If $(p, A, f)$ and $(q, A, b)$ are two configurations, $(p, A, f)$ $~\vdash_M$ $(q, A, b)$ or $(p, A, b)$ $~\vdash_M$ $(q, A, f)$ if $p\# \rightarrow q \in R_{\#}$.\\

The reflexive transitive closure of $~\vdash_M$ is denoted by $~\vdash*{_M}$. A picture $A\in \Sigma^{++H}$ is accepted by a GHBFA $M$ with direction $d_{HBFA}$~=~${\begin{pmatrix}
 s\downarrow&\CircleArrowright\\
\CircleArrowleft&\uparrow
\end{pmatrix}}$ if $c_{init}(A)$$~\vdash*{_M}$~$c$ where $c$ is a final configuration.\\

Figure \ref{working} illustrates the scanning process of GHBFA on an input image,  the sequence of $\square$s indicating the progress made in processing the input. 
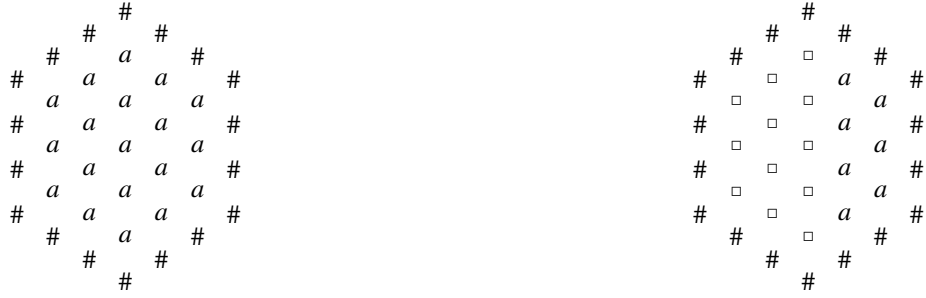
\begin{figure}[h]
    \centering
    
    % First subfigure
    \begin{minipage}{0.45\textwidth}
        \centering
        \begin{tikzpicture}[rotate=90, scale=0.6]
        \def\n{3}
        \foreach \i in {0,...,3} {
            \pgfmathsetmacro{\num}{\n + \i}
            \pgfmathsetmacro{\startx}{-\i*0.5}
            \foreach \j in {0,...,\num} {
                \ifnum\i=0
                    \node at ({\startx + \j}, {-0.8*\i}) {\#};
                \else
                    \ifnum\j=0
                        \node at ({\startx + \j}, {-0.8*\i}) {\#};
                    \else
                        \ifnum\j=\num
                            \node at ({\startx + \j}, {-0.8*\i}) {\#};
                        \else
                            \node at ({\startx + \j}, {-0.8*\i}) {\textit{a}};
                        \fi
                    \fi
                \fi
            }
        }
        \foreach \i in {1,...,3} {
            \pgfmathsetmacro{\num}{\n + 3 - \i}
            \pgfmathsetmacro{\startx}{- (3 - \i)*0.5}
            \foreach \j in {0,...,\num} {
                \ifnum\i=3
                    \node at ({\startx + \j}, {-0.8*(\i + 3)}) {\#};
                \else
                    \ifnum\j=0
                        \node at ({\startx + \j}, {-0.8*(\i + 3)}) {\#};
                    \else
                        \ifnum\j=\num
                            \node at ({\startx + \j}, {-0.8*(\i + 3)}) {\#};
                        \else
                            \node at ({\startx + \j}, {-0.8*(\i + 3)}) {\textit{a}};
                        \fi
                    \fi
                \fi
            }
        }
        \end{tikzpicture}
    \end{minipage}
    \hfill
    % Second subfigure
    \begin{minipage}{0.45\textwidth}
        \centering
        \begin{tikzpicture}[rotate=90, scale=0.6]
        \def\n{3}
        \foreach \i in {0,...,3} {
            \pgfmathsetmacro{\num}{\n + \i}
            \pgfmathsetmacro{\startx}{-\i*0.5}
            \foreach \j in {0,...,\num} {
                \ifnum\i=0
                    \node at ({\startx + \j}, {-0.8*\i}) {\#};
                \else
                    \ifnum\j=0
                        \node at ({\startx + \j}, {-0.8*\i}) {\#};
                    \else
                        \ifnum\j=\num
                            \node at ({\startx + \j}, {-0.8*\i}) {\#};
                        \else
                        \node at ({\startx + \j}, {-0.8*\i}) {\textit{$\square$}};
                        \fi
                    \fi
                \fi
            }
        }
        \foreach \i in {1,...,3} {
            \pgfmathsetmacro{\num}{\n + 3 - \i}
            \pgfmathsetmacro{\startx}{- (3 - \i)*0.5}
            \foreach \j in {0,...,\num} {
                \ifnum\i=3
                    \node at ({\startx + \j}, {-0.8*(\i + 3)}) {\#};
                \else
                    \ifnum\j=0
                        \node at ({\startx + \j}, {-0.8*(\i + 3)}) {\#};
                    \else
                        \ifnum\j=\num
                            \node at ({\startx + \j}, {-0.8*(\i + 3)}) {\#};
                        \else
                            \node at ({\startx + \j}, {-0.8*(\i + 3)}) {\textit{a}};
                        \fi
                    \fi
                \fi
            }
        }
        \end{tikzpicture}
    \end{minipage}
    \caption{a)Input Picture~~b)Processing of (a) using  direction ~${\begin{pmatrix}
 s\downarrow&\CircleArrowright\\
\CircleArrowleft&\uparrow
\end{pmatrix}}$  }
\label{working}
\end{figure}
encapsulates all the information required to describes a configuration excluding the states.
 Now we have \\

 \begin{tabular}{|c|c|c|c|c|c|c|}
 \hline
$d=$&$\begin{pmatrix}
 s\downarrow&\CircleArrowright\\
\CircleArrowleft&\uparrow
\end{pmatrix}$&$\begin{pmatrix}
 s\rotatebox[origin=c]{300}{$\uparrow$}&\CircleArrowleft\\
\CircleArrowright&\rotatebox[origin=c]{300}{$\downarrow$}
\end{pmatrix}$&$\begin{pmatrix}
\CircleArrowright &\rotatebox[origin=c]{60}{$\uparrow$}\\
s\rotatebox[origin=c]{60}{$\downarrow$}&\CircleArrowleft
\end{pmatrix}$&$\begin{pmatrix}
\CircleArrowleft &\downarrow\\
s\uparrow&\CircleArrowright
\end{pmatrix}$&$\begin{pmatrix}
\CircleArrowright &\rotatebox[origin=c]{-60}{$\downarrow$}\\
s\rotatebox[origin=c]{300}{$\uparrow$}&\CircleArrowleft
\end{pmatrix}$&$\begin{pmatrix}
 \rotatebox[origin=c]{60}{$\downarrow$}& \CircleArrowright\\
\CircleArrowleft&s\rotatebox[origin=c]{60}{$\uparrow$}
\end{pmatrix}$\\
\hline
$f_d(A)$&$R_0$&$r_5$&$R_1$&$r_0$&$R_2$&$r_1$\\
\hline
$d=$&$\begin{pmatrix}
 \rotatebox[origin=c]{300}{$\uparrow$}&\CircleArrowright\\
\CircleArrowleft&s\rotatebox[origin=c]{-60}{$\downarrow$}
\end{pmatrix}$&$\begin{pmatrix}
 \downarrow&\CircleArrowleft\\
\CircleArrowright&s\uparrow
\end{pmatrix}$&$\begin{pmatrix}
\CircleArrowleft&s\rotatebox[origin=c]{60}{$\uparrow$}\\
\rotatebox[origin=c]{-60}{$\rightarrow$}&\CircleArrowright 
\end{pmatrix}$&$\begin{pmatrix}
\CircleArrowright &s\downarrow\\
\uparrow&\CircleArrowleft
\end{pmatrix}$&$\begin{pmatrix}
s\rotatebox[origin=c]{60}{$\leftarrow$}&\CircleArrowright \\
\CircleArrowleft & \rotatebox[origin=c]{300}{$\uparrow$}
\end{pmatrix}$&$\begin{pmatrix}
s \rotatebox[origin=c]{-60}{$\rightarrow$} & \CircleArrowleft\\
\CircleArrowright& \rotatebox[origin=c]{60}{$\uparrow$}
\end{pmatrix}$\\
\hline
$f_d(A)$&$r_2$&$R_3$&$R_4$&$r_3$&$R_5$&$r_4$\\
\hline
\end{tabular}\\

A picture $A$ is accepted by a $GHBFA$ $M$ with a direction $d$ if and only if $f_d(A)$ is accepted by the $GHBFA$ $M_{HBFA}$ that coincides with $M$ in every detail except for the direction $d$. For example, if  we consider
$d= \begin{pmatrix}
\CircleArrowright &s\downarrow\\
\uparrow&\CircleArrowleft
\end{pmatrix}$, in order to scan the given input picture $A$ we will start from the top right corner going downwards then upwards, instead of this we can take the reflection of given picture along the line of symmetry by the angle $90^{\circ}$ and start from the left top corner by reading top-down then bottom - up.
\\

The $GHBFA$ is deterministic $(GHBDFA)$, if for each $p\in Q$ and $a \in \Sigma \cup \{\#\}$, there is at most one $q \in Q$ with $pa \rightarrow q\in R$. Now  we define $L_d(GHBFA)$ as the language accepted by $GHBFA$ when working with the direction $d$. Then,
\[L(GHBFA) = \underset{d\in D}{\bigcup}L_d(GHBFA)\]
The following characterization is an immediate consequence of the definitions.
\begin{theorem}
The class $L_{HBFA}(GHBFA)$ coincides with the following classes of picture languages : 

$r_5(L_{\tiny {\begin{pmatrix}
 s\rotatebox[origin=c]{300}{$\uparrow$}&\CircleArrowleft\\
\CircleArrowright&\rotatebox[origin=c]{300}{$\downarrow$}
\end{pmatrix}}}
(GHBFA))$, $R_1(L_{\tiny{\begin{pmatrix}
\CircleArrowright &\rotatebox[origin=c]{60}{$\uparrow$}\\
s\rotatebox[origin=c]{60}{$\downarrow$}&\CircleArrowleft
\end{pmatrix}}}(GHBFA))$, $r_0(L_{\tiny{\begin{pmatrix}
\CircleArrowleft &\downarrow\\
s\uparrow&\CircleArrowright
\end{pmatrix}}}(GHBFA))$, $R_2(L_{\tiny{\begin{pmatrix}
\CircleArrowright &\rotatebox[origin=c]{-60}{$\downarrow$}\\
s\rotatebox[origin=c]{300}{$\uparrow$}&\CircleArrowleft
\end{pmatrix}}}(GHBFA))$, $r_1(L_{\tiny{\begin{pmatrix}
 \rotatebox[origin=c]{60}{$\downarrow$}& \CircleArrowright\\
\CircleArrowleft&s\rotatebox[origin=c]{60}{$\uparrow$}
\end{pmatrix}}}(GHBFA))$, $r_2(L_{\tiny{\begin{pmatrix}
 \rotatebox[origin=c]{300}{$\uparrow$}&\CircleArrowright\\
\CircleArrowleft&s\rotatebox[origin=c]{-60}{$\downarrow$}
\end{pmatrix}}}(GHBFA))$, $R_3(L_{\tiny{\begin{pmatrix}
 \downarrow&\CircleArrowleft\\
\CircleArrowright&s\uparrow
\end{pmatrix}}}(GHBFA))$, $R_4(L_{\tiny{\begin{pmatrix}
\CircleArrowleft&s\rotatebox[origin=c]{60}{$\uparrow$}\\
\rotatebox[origin=c]{-60}{$\rightarrow$}&\CircleArrowright 
\end{pmatrix}}}(GHBFA))$, $r_3(L_{\tiny{\begin{pmatrix}
\CircleArrowright &s\downarrow\\
\uparrow&\CircleArrowleft
\end{pmatrix}}}(GHBFA))$, $R_5(L_{\tiny{\begin{pmatrix}
s\rotatebox[origin=c]{60}{$\leftarrow$}&\CircleArrowright \\
\CircleArrowleft & \rotatebox[origin=c]{300}{$\uparrow$}
\end{pmatrix}}}(GHBFA))$, $r_4(L_{\tiny{\begin{pmatrix}
s \rotatebox[origin=c]{-60}{$\rightarrow$} & \CircleArrowleft\\
\CircleArrowright& \rotatebox[origin=c]{60}{$\uparrow$}
\end{pmatrix}}}(GHBFA))$.
\end{theorem}
The group-theoretic framework above enables us to derive the characterizations of the remaining eleven classes from the previous theorem referring back to $L_{d_{HBFA}}(GHBFA)$. These results are combined in the next theorem.
\begin{theorem}
The resulting characterizations are listed~ below
\begin{itemize}
\item[] $L_{\tiny{\begin{pmatrix}
 s\rotatebox[origin=c]{300}{$\uparrow$}&\CircleArrowleft\\
\CircleArrowright&\rotatebox[origin=c]{300}{$\downarrow$}
\end{pmatrix}}}(GHBFA) = r_1\circ (R_1\circ R_1)(L_{d_{HBFA}}(GHBFA))$\\
\item[] $L_{\tiny{\begin{pmatrix}
\CircleArrowright&\rotatebox[origin=c]{60}{$\uparrow$}\\
s\rotatebox[origin=c]{60}{$\downarrow$}&\CircleArrowleft
\end{pmatrix}}}(GHBFA) = R_1(L_{d_{HBFA}}(GHBFA))$ \\ 
\item[]$L_{\tiny{\begin{pmatrix}
\CircleArrowleft &\downarrow\\
s\uparrow&\CircleArrowright
\end{pmatrix}}}(GHBFA) = r_1\circ R_1(L_{d_{HBFA}}(GHBFA))$\\
\item[] $L_{\tiny{\begin{pmatrix}
\CircleArrowright &\rotatebox[origin=c]{-60}{$\downarrow$}\\
s\rotatebox[origin=c]{300}{$\uparrow$}&\CircleArrowleft
\end{pmatrix}}}(GHBFA) = R_1\circ R_1(L_{d_{HBFA}}(GHBFA))$\\
\item[] $L_{\tiny{\begin{pmatrix}
 \rotatebox[origin=c]{60}{$\downarrow$}& \CircleArrowright\\
\CircleArrowleft&s\rotatebox[origin=c]{60}{$\uparrow$}
\end{pmatrix}}}(GHBFA) = r_1(L_{d_{HBFA}}(GHBFA))$\\
\item[] $L_{\tiny{\begin{pmatrix}
 \rotatebox[origin=c]{300}{$\uparrow$}&\CircleArrowright\\
\CircleArrowleft&s\rotatebox[origin=c]{-60}{$\downarrow$}
\end{pmatrix}}}(GHBFA) = r_1\circ (R_1\circ((R_1\circ R_1)\circ(R_1\circ R_1))) (L_{d_{HBFA}}(GHBFA))$\\
\item[]$L_{\tiny{\begin{pmatrix}
 \downarrow&\CircleArrowleft\\
\CircleArrowright&s\uparrow
\end{pmatrix}}}(GHBFA) = R_1\circ(R_1\circ R_1)(L_{d_{HBFA}}(GHBFA))$\\
\item[]$L_{\tiny{\begin{pmatrix}
\CircleArrowleft&s\rotatebox[origin=c]{60}{$\uparrow$}\\
\rotatebox[origin=c]{-60}{$\rightarrow$}&\CircleArrowright 
\end{pmatrix}}}(GHBFA) = (R_1\circ R_1)\circ(R_1\circ R_1)(L_{d_{HBFA}}(GHBFA))$\\
\item[]$L_{\tiny{\begin{pmatrix}
\CircleArrowright &s\downarrow\\
\uparrow&\CircleArrowleft
\end{pmatrix}}}(GHBFA) = r_1\circ ((R_1\circ R_1)\circ(R_1\circ R_1))(L_{d_{HBFA}}(GHBFA))$\\
\item[]$L_{\tiny{\begin{pmatrix}
s\rotatebox[origin=c]{60}{$\leftarrow$}&\CircleArrowright \\
\CircleArrowleft & \rotatebox[origin=c]{300}{$\uparrow$}
\end{pmatrix}}}(GHBFA) = R_1\circ((R_1\circ R_1)\circ(R_1\circ R_1))(L_{d_{HBFA}}(GHBFA))$\\
\item[]$L_{\tiny{\begin{pmatrix}
s \rotatebox[origin=c]{-60}{$\rightarrow$} & \CircleArrowleft\\
\CircleArrowright& \rotatebox[origin=c]{60}{$\uparrow$}
\end{pmatrix}}}(GHBFA) =r_1\circ  (R_1\circ(R_1\circ R_1)) (L_{d_{HBFA}}(GHBFA))$\\
\end{itemize}
These~characterizations~are~also~ valid ~for~ the~ corresponding~deterministic~classes.
\end{theorem}
\begin{theorem}
For~ each~ direction ~mode $d$\[L_d(GHBFA) = L_d(GHBDFA)\]   
\end{theorem}
\begin{proof}
  Since our GHBFAs are syntactically identical to classical finite automata except in the way they are processed, we can apply the  subset construction to establish the result.
\end{proof}
\section{General Hexagonal Returning Finite Automata}\label{GHRFA}
Fernau et al. \cite{fernau2015scanning} introduces Returning Finite Automata(RFA) working on rectangular pictures  which differ from Boustrophedon Finite Automata only in the direction of row processing - left to right, followed by top to bottom.  We adapt this definition to introduce Hexagonal Returning Finite Automata, which uses only the directions $d$~=~${\begin{pmatrix}
 s\downarrow&\CircleArrowright
\end{pmatrix}}$ which indicates that  we  start from the top left corner of the hexagon and  process the picture in top down manner, then  go to the next column using \CircleArrowright ~again continuing in top down manner.\\ General Hexagonal Returning Finite Automata, have also been introduced as in \cite{picturescanning} where the set of directions is given by 
\[D' =\Big\{\begin{pmatrix}
s \rotatebox[origin=c]{-60}{$\rightarrow$}& \CircleArrowleft \\
\end{pmatrix}, \begin{pmatrix}
s  \rotatebox[origin=c]{60}{$\leftarrow$}& \CircleArrowright \\
\end{pmatrix}, \begin{pmatrix}
\CircleArrowleft&s\downarrow\\
 \end{pmatrix}, \begin{pmatrix}
\CircleArrowright&s\rotatebox[origin=c]{60}{$\uparrow$}\\
\end{pmatrix}, \begin{pmatrix}
 \CircleArrowright&s\uparrow
\end{pmatrix}, \begin{pmatrix}
 \CircleArrowleft&s\rotatebox[origin=c]{-60}{$\downarrow$}
\end{pmatrix},\]
\[\begin{pmatrix}
\CircleArrowleft&s\rotatebox[origin=c]{60}{$\uparrow$}
\end{pmatrix}, \begin{pmatrix}
s\rotatebox[origin=c]{300}{$\uparrow$}&\CircleArrowright
\end{pmatrix}, \begin{pmatrix}
 s\uparrow&\CircleArrowleft 
\end{pmatrix}, \begin{pmatrix}
s\rotatebox[origin=c]{60}{$\downarrow$}&\CircleArrowright
\end{pmatrix}, \begin{pmatrix}
 s\rotatebox[origin=c]{300}{$\uparrow$}&\CircleArrowleft\\
\end{pmatrix}, \begin{pmatrix}
 s\downarrow&\CircleArrowright\\
\end{pmatrix}
\Big\}\]
\par 
We define $L_d(GHRFA) $ as the language accepted by $GHRFA$ when working with the direction $d$. The  set $D'$ can be partitioned into 3 sets $D_1$, $D_2$, $D_3$, where\\
$D_1 = { \begin{pmatrix}
\CircleArrowleft&s\downarrow\\
 \end{pmatrix}, \begin{pmatrix}
 \CircleArrowright&s\uparrow
\end{pmatrix}, \begin{pmatrix}
 s\uparrow&\CircleArrowleft 
\end{pmatrix}, \begin{pmatrix}
 s\downarrow&\CircleArrowright\\
\end{pmatrix}}$\\
$D_2= {\begin{pmatrix}
s  \rotatebox[origin=c]{-60}{$\rightarrow$} &\CircleArrowleft\\
\end{pmatrix}, \begin{pmatrix}
 \CircleArrowright&s\rotatebox[origin=c]{60}{$\uparrow$}\\
\end{pmatrix}, \begin{pmatrix}
\CircleArrowleft&s\rotatebox[origin=c]{60}{$\uparrow$}
\end{pmatrix}, \begin{pmatrix}
s\rotatebox[origin=c]{60}{$\downarrow$}&\CircleArrowright
    \end{pmatrix}}$\\
$D_3 = {\begin{pmatrix}
s  \rotatebox[origin=c]{60}{$\leftarrow$}&\CircleArrowright \\
\end{pmatrix}, \begin{pmatrix}
\CircleArrowleft&s\rotatebox[origin=c]{-60}{$\downarrow$}
\end{pmatrix}, \begin{pmatrix}
s\rotatebox[origin=c]{300}{$\uparrow$}&\CircleArrowright
\end{pmatrix}, \begin{pmatrix}
 s\rotatebox[origin=c]{300}{$\uparrow$}&\CircleArrowleft\\
\end{pmatrix}}$\\It was proved in \cite{fernau2015scanning} that $L(BFA)=L(RFA)$. In a similar manner we can prove  
\begin{theorem}
   $L(HBFA)=L(HRFA)$ 
\end{theorem}
\begin{proof}
   The states of the $HBFA$ are triples $(p, q, r)$, where $p$ is the actual
 state, $q$ is the state that should be reached after finishing a column and $r$ indicates
 whether or not a top to bottom  (denoted as 1) or a bottom to top scan (denoted as 2) is performed. The formal definition is as follows.
\par  Let $M = (Q, \Sigma, R, s, F, \#, \square, D)$ be some $HBFA$. We define the equivalent $HRFA$
$M' = (Q', \Sigma, R', s', F', \#, \square, D')$ as follows: 
\begin{align*}
Q'&= Q \times Q \times \{1,2\} \cup \{s'\} \\
R'& =\{(p,r,1)a\rightarrow (q,r,1)|~ pa\rightarrow q \in R,r\in Q\}\\&\cup\{(q,r,2)\rightarrow (p,r,2)|~ pa \rightarrow q \in R, r\in Q\}\\&\cup\{(p,p,1)\# \rightarrow (r,q,2)|~ p\# \rightarrow q \in R, r\in Q\} \\&\cup \{(p, p, 2)\# \rightarrow (r, q, 1)|~ p\# \rightarrow q \in R, r\in Q\} \\&\cup\{s'a \rightarrow (p, r, 1)|~ s'a \rightarrow p\in R, r\in Q\}\\
F' &= \{(r, r, 1)|r\in F'\}\cup \{(p, r, 2)|r \in F', p \in Q\}.
\end{align*}

It can be easily proved that the languages  of $M$ and $M'$ are the same. In other words  $L_{d_{HBFA}}(GHBFA) = L_{d_{HRFA}}(GHRFA)$.
\end{proof}
\begin{theorem}
    Let $d_{HRFA}$ =$\begin{pmatrix}
 s\downarrow&\CircleArrowright
\end{pmatrix}$ then $L_{d_{HRFA}}$ coincide with the following classes:\\
$r_5(L_{\tiny{\begin{pmatrix}
 s\rotatebox[origin=c]{300}{$\uparrow$}&\CircleArrowleft\\
\end{pmatrix}}}(GHRFA))$, $R_1(L_{\tiny{\begin{pmatrix}
s\rotatebox[origin=c]{60}{$\downarrow$}&\CircleArrowright
\end{pmatrix}}}(GHRFA))$, $r_0(L_{\tiny{\begin{pmatrix}
 s\uparrow&\CircleArrowleft 
\end{pmatrix}}}(GHRFA))$,\\ $R_2(L_{\tiny{\begin{pmatrix}
s\rotatebox[origin=c]{300}{$\uparrow$}&\CircleArrowright
\end{pmatrix}}}(GHRFA))$,$r_1(L_{\tiny{\begin{pmatrix}
\CircleArrowleft&s\rotatebox[origin=c]{60}{$\uparrow$}
\end{pmatrix}}}(GHRFA))$, $r_2(L_{\tiny{\begin{pmatrix}
 \CircleArrowleft&s\rotatebox[origin=c]{-60}{$\downarrow$}
\end{pmatrix}}}(GHRFA))$,\\ $R_3(L_{\tiny{\begin{pmatrix}
 \CircleArrowright&s\uparrow
\end{pmatrix}}}(GHRFA))$, $R_4(L_{\tiny{\begin{pmatrix}
\CircleArrowright&s\rotatebox[origin=c]{60}{$\uparrow$}\\
\end{pmatrix}}}(GHRFA))$, $r_3(L_{\tiny{\begin{pmatrix}
\CircleArrowleft&s\downarrow\\
 \end{pmatrix}}}(GHRFA))$,\\ $R_5(L_{\tiny{\begin{pmatrix}
s  \rotatebox[origin=c]{60}{$\leftarrow$}& \CircleArrowright \\
\end{pmatrix}}}(GHRFA))$, $r_4(L_{\tiny{\begin{pmatrix}
s \rotatebox[origin=c]{-60}{$\rightarrow$} & \CircleArrowleft \\
\end{pmatrix}}}(GHRFA))$.
\end{theorem}
Using the previous theorem, we can obtain a representation similar to Theorem 2 for $GHRFA$
\begin{theorem}
   The list of characterizations is as follows: \\
\begin{itemize}
\item[] $L_{\tiny{\begin{pmatrix}
 s\rotatebox[origin=c]{300}{$\uparrow$}&\CircleArrowleft\\
\end{pmatrix}}}(GHRFA) = r_1\circ (R_1\circ R_1)(L_{d_{HRFA}}(GHRFA))$\\
\item[] $L_{\tiny{\begin{pmatrix}
s\rotatebox[origin=c]{60}{$\downarrow$}&\CircleArrowright
\end{pmatrix}}}(GHRFA) = R_1(L_{d_{HRFA}}(GHRFA))$ \\ 
\item[]$L_{\tiny{\begin{pmatrix}
 s\uparrow&\CircleArrowleft 
\end{pmatrix}}}(GHRFA) = r_1\circ R_1(L_{d_{HRFA}}(GHRFA))$\\
\item[] $L_{\tiny{\begin{pmatrix}
s\rotatebox[origin=c]{300}{$\uparrow$}&\CircleArrowright
\end{pmatrix}}}(GHRFA) = R_1\circ R_1(L_{d_{HRFA}}(GHRFA))$\\
\item[] $L_{\tiny{\begin{pmatrix}
\CircleArrowleft&s\rotatebox[origin=c]{60}{$\uparrow$}
\end{pmatrix}}}(GHRFA) = r_1(L_{d_{HRFA}}(GHRFA))$\\
\item[] $L_{\tiny{\begin{pmatrix}
 \CircleArrowleft&s\rotatebox[origin=c]{-60}{$\downarrow$}
\end{pmatrix}}}(GHRFA) = r_1\circ (R_1\circ(R_1\circ R_1)\circ(R_1\circ R_1)) (L_{d_{HRFA}}(GHRFA))$\\
\item[]$L_{\tiny{\begin{pmatrix}
 \CircleArrowright&s\uparrow
\end{pmatrix}}}(GHRFA) = R_1\circ(R_1\circ R_1)(L_{d_{HRFA}}(GHRFA))$\\
\item[]$L_{\tiny{\begin{pmatrix}
\CircleArrowright&s\rotatebox[origin=c]{60}{$\uparrow$}\\
\end{pmatrix}}}(GHRFA) = (R_1\circ R_1)\circ(R_1\circ R_1)(L_{d_{HRFA}}(GHRFA))$\\
\item[]$L_{\tiny{\begin{pmatrix}
\CircleArrowleft&s\downarrow\\
 \end{pmatrix}}}(GHRFA) = r_1\circ ((R_1\circ R_1)\circ(R_1\circ R_1))(L_{d_{HRFA}}(GHRFA))$\\
\item[]$L_{\tiny{\begin{pmatrix}
s  \rotatebox[origin=c]{60}{$\leftarrow$}& \CircleArrowright \\
\end{pmatrix}}}(GHRFA) = R_1\circ((R_1\circ R_1)\circ(R_1\circ R_1))(L_{d_{HRFA}}(GHRFA))$\\
\item[]$L_{\tiny{\begin{pmatrix}
s \rotatebox[origin=c]{-60}{$\rightarrow$} & \CircleArrowleft \\
\end{pmatrix}}}(GHRFA) =r_1\circ  (R_1\circ(R_1\circ R_1)) (L_{d_{HRFA}}(GHRFA))$\\
\end{itemize}

\end{theorem}
\section{Equivalences of Language Families }\label{equi}
\begin{lemma}
    $L_d(GHRFA)= r_0(L_d(GHRFA))$ for $d \in D_1$
\end{lemma}
\begin{proof}
    Let $M= (Q, \Sigma, R, s, F, \#, \square, \begin{pmatrix}
 s\downarrow&\CircleArrowright
\end{pmatrix})$ be some $GHRFA$ and let $L=L(M)$. Let us construct some $GHRFA$ $M_h$ that accepts $r_0(L)$. \par Define $M_h=(Q_h, \Sigma, R_h, Q_I, Q_F, \#, \square, \begin{pmatrix}
 s\downarrow&\CircleArrowright
\end{pmatrix})$ where 
\begin{align*}
   Q_h&= Q\times Q\times Q,\\&Q_\#= \{q |~ pa \rightarrow q~\Lambda~q\# \rightarrow r \in R\}, \\&Q_I = \{(s, q, r)\}~ where ~q\in Q, r \in Q_\#, \\& R_h = \{(l, q, r)a \rightarrow(l, p, r) |~pa \rightarrow q \in R, l \in Q, r \in Q_\#, a \in \Sigma\}\\&\cup \{(l, l, p)a \rightarrow(q, t, r) |~p\# \rightarrow q \in R, l, r, t \in Q, a \in \Sigma\}\\&\cup \{(l, q, f)a \rightarrow(l, p, f) |~pa \rightarrow q \in R, l\in Q, f\in F, a \in \Sigma\}~and\\& Q_F \subset Q_h, Q_F = \{(l, l, f)~|~f\in F, l \in Q\}.
 \end{align*}\\
 The core idea behind the construction is analogous to the mirror-image technique familiar from classical formal language theory. Specifically, in a triple $(l, q, r) \in Q \times Q \times Q$. Here, the the first component stores the state corresponding to the topmost symbol of the current column, $q$ represents the current state, and $r$ is associated to the bottom most symbol in the column.
  \end{proof}
  Since horizontal reflection can similarly be interpreted as a change in the processing mode, the corresponding result follows immediately.
  \begin{cor}
$L_{\tiny{\begin{pmatrix}
 s\downarrow&\CircleArrowright
\end{pmatrix}}}(GHRFA) =L_{\tiny{\begin{pmatrix}
 s\uparrow&\CircleArrowleft 
\end{pmatrix}}}(GHRFA) $ and\\ $L_{\tiny{\begin{pmatrix}
\CircleArrowleft &s\downarrow\\
 \end{pmatrix}}}(GHRFA) =L_{\tiny{\begin{pmatrix}
 \CircleArrowright&s\uparrow
\end{pmatrix}}}(GHRFA) $ 
  \end{cor}
  By Theorem 6 and using a construction similar to that used in Lemma 1 we may also conclude that
  \begin{cor}
 $L_d(GHRFA)= r_3(L_d(GHRFA))$ for $d \in D_v$
  \end{cor}Accordingly, we obtain the following characterization, which follows directly from the previous result:
  \begin{cor}
    $L_{\tiny{\begin{pmatrix}
 s\downarrow&\CircleArrowright
\end{pmatrix}}}(GHRFA) =L_{\tiny{\begin{pmatrix}
\CircleArrowleft &s\downarrow\\
 \end{pmatrix}}}(GHRFA) $ and\\ $L_{\tiny{\begin{pmatrix}
 s\uparrow&\CircleArrowleft 
\end{pmatrix}}}(GHRFA) =L_{\tiny{\begin{pmatrix}
 \CircleArrowright&s\uparrow
\end{pmatrix}}}(GHRFA) $  
  \end{cor}
  Using Theorem 4 and Lemma 1 we can conclude that 
\begin{cor}
 $L_{d_{HBFA}}(GHBFA)= r_0(L_{d_{HBFA}}(GHBFA))$
  \end{cor}
   Using Theorem 4 and corollary 2 we can conclude that 
   \begin{cor}
   $L_{d_{HBFA}}(GHBFA)= r_3(L_{d_{HBFA}}(GHBFA))$
   \end{cor}
   As $r_0 \circ r_3 = R_3$ Corollary 4 and Corollary 5 give,
\begin{cor}
$L_{d_{HBFA}}(GHBFA)= R_3(L_{d_{HBFA}}(GHBFA))$
   \end{cor}
\begin{theorem}
    The picture language family $L_{d_{HBFA}}(GHBFA)$ equals to $L_d(GHBFA)$ for $d \in \Big\{d_{HBFA}, \begin{pmatrix}
\CircleArrowleft &\downarrow\\
s\uparrow&\CircleArrowright
\end{pmatrix}, \begin{pmatrix}
\CircleArrowright &s\downarrow\\
\uparrow&\CircleArrowleft
\end{pmatrix}, \begin{pmatrix}
 \downarrow&\CircleArrowleft\\
\CircleArrowright&s\uparrow
\end{pmatrix}$\Big\}
\end{theorem}
\begin{proof}
Using Corollary 4 and Theorem 2\\
$L_{d_{HBFA}}(GHBFA) = r_0(L_{d_{HBFA}}(GHBFA)) = L_{\tiny{\begin{pmatrix}
\CircleArrowleft &\downarrow\\
s\uparrow&\CircleArrowright
\end{pmatrix}}}(GHBFA)$\\Using Corollary 5 and Theorem 2\\
$L_{d_{HBFA}}(GHBFA) = r_3(L_{d_{HBFA}}(GHBFA)) = L_{\tiny{\begin{pmatrix}
\CircleArrowright &s\downarrow\\
\uparrow&\CircleArrowleft
\end{pmatrix}}}(GHBFA)$\\Using Corollary 6 and Theorem 2\\
$L_{d_{HBFA}}(GHBFA) = R_3(L_{d_{HBFA}}(GHBFA)) = L_{\tiny{\begin{pmatrix}
 \downarrow&\CircleArrowleft\\
\CircleArrowright&s\uparrow
\end{pmatrix}}}(GHBFA)$
\end{proof}
    \begin{theorem}
 The picture language family $L_{d_{HBFA}}(GHBFA)$ equals to $L_d(GHRFA)$ for $d \in D_1$      
    \end{theorem}
    \begin{proof}
        Using Theorem 4, Corollary  1 and Corollary 3,\\ $L_{d_{HBFA}}(GHBFA)$ = $L_{d_{HRFA}}(GHRFA)$ = $L_{\tiny{\begin{pmatrix}
 s\downarrow&\CircleArrowright
\end{pmatrix}}}(GHRFA) =L_{\tiny{\begin{pmatrix}
 s\uparrow&\CircleArrowleft 
\end{pmatrix}}}(GHRFA) $ =$L_{\tiny{\begin{pmatrix}
 \CircleArrowright&s\uparrow
\end{pmatrix}}}(GHRFA)$ =$L_{\tiny{\begin{pmatrix}
\CircleArrowleft &s\downarrow\\
 \end{pmatrix}}}(GHRFA)$.
    \end{proof}By Theorem 6 and using a construction similar to that used in Lemma 1
we may also conclude that
\begin{cor}
  $L_d(GHRFA)= r_2(L_d(GHRFA))$ for $d \in D_M$
\end{cor}
\begin{cor}
$L_d(GHRFA)= r_4(L_d(GHRFA))$ for $d \in D_M$
\end{cor}
Using Corollary 7
\begin{cor}
    $L_{\tiny{\begin{pmatrix}
s  \rotatebox[origin=c]{-60}{$\rightarrow$} &\CircleArrowleft\\
\end{pmatrix}}}(GHRFA) =L_{\tiny{\begin{pmatrix}
 \CircleArrowright&s\rotatebox[origin=c]{60}{$\uparrow$}\\
\end{pmatrix}}}(GHRFA) $ and\\ $L_{\tiny{\begin{pmatrix}
\CircleArrowleft&s\rotatebox[origin=c]{60}{$\uparrow$}
\end{pmatrix}}}(GHRFA) =L_{\tiny{\begin{pmatrix}
s\rotatebox[origin=c]{60}{$\downarrow$}&\CircleArrowright
    \end{pmatrix}}}(GHRFA) $ 
\end{cor}
Using Corollary 8
\begin{cor}
$L_{\tiny{\begin{pmatrix}
s  \rotatebox[origin=c]{-60}{$\rightarrow$} &\CircleArrowleft\\
\end{pmatrix}}}(GHRFA) = L_{\tiny{\begin{pmatrix}
s\rotatebox[origin=c]{60}{$\downarrow$}&\CircleArrowright
    \end{pmatrix}}}(GHRFA)$ and\\$L_{\tiny{\begin{pmatrix}
 \CircleArrowright&s\rotatebox[origin=c]{60}{$\uparrow$}\\
\end{pmatrix}}}(GHRFA) = L_{\tiny{\begin{pmatrix}
\CircleArrowleft&s\rotatebox[origin=c]{60}{$\uparrow$}
\end{pmatrix}}}(GHRFA) $
\end{cor}
\begin{theorem}
  The picture language family $r_1(L_{d_{HBFA}}(GHBFA))$ equals to $L_d(GHBFA)$ for $d \in \Big\{\begin{pmatrix}
 \rotatebox[origin=c]{60}{$\downarrow$}& \CircleArrowright\\
\CircleArrowleft&s\rotatebox[origin=c]{60}{$\uparrow$}
\end{pmatrix}, \begin{pmatrix}
s \rotatebox[origin=c]{-60}{$\rightarrow$} & \CircleArrowleft\\
\CircleArrowright& \rotatebox[origin=c]{60}{$\uparrow$}
\end{pmatrix}, \begin{pmatrix}
\CircleArrowleft&s\rotatebox[origin=c]{60}{$\uparrow$}\\
\rotatebox[origin=c]{-60}{$\rightarrow$}&\CircleArrowright 
\end{pmatrix}, \begin{pmatrix}
\CircleArrowright &\rotatebox[origin=c]{60}{$\uparrow$}\\
s\rotatebox[origin=c]{60}{$\downarrow$}&\CircleArrowleft
\end{pmatrix}\Big\}  $
\end{theorem}
\begin{proof}
  $r_1(L_{d_{HBFA}}(GHBFA))$ =  $L_{\tiny{\begin{pmatrix}
 \rotatebox[origin=c]{60}{$\downarrow$}& \CircleArrowright\\
\CircleArrowleft&s\rotatebox[origin=c]{60}{$\uparrow$}
\end{pmatrix}}}(GHBFA) $\\
$r_1(L_{d_{HBFA}}(GHBFA))$ =$r_4(R_3(L_{d_{HBFA}}(GHBFA)))$ = $r_4(L_{d_{HBFA}}(GHBFA))$ = $L_{\tiny{\begin{pmatrix}
s \rotatebox[origin=c]{-60}{$\rightarrow$} & \CircleArrowleft\\
\CircleArrowright& \rotatebox[origin=c]{60}{$\uparrow$}
\end{pmatrix}}}(GHBFA)$(Using Table 1, Corollary 6 and Theorem 2 )\\
$r_1(L_{d_{HBFA}}(GHBFA))$ =$R_4(r_3(L_{d_{HBFA}}(GHBFA)))$ = $R_4(L_{d_{HBFA}}(GHBFA))$ = $L_{\tiny{\begin{pmatrix}
\CircleArrowleft&s\rotatebox[origin=c]{60}{$\uparrow$}\\
\rotatebox[origin=c]{-60}{$\rightarrow$}&\CircleArrowright 
\end{pmatrix}}}(GHBFA)$(Using Table 1, Corollary 5 and Theorem 2 )\\
$r_1(L_{d_{HBFA}}(GHBFA))$ =$R_1(r_0(L_{d_{HBFA}}(GHBFA)))$ = $R_1(L_{d_{HBFA}}(GHBFA))$ = $L_{\tiny{\begin{pmatrix}
\CircleArrowright &\rotatebox[origin=c]{60}{$\uparrow$}\\
s\rotatebox[origin=c]{60}{$\downarrow$}&\CircleArrowleft
\end{pmatrix}}}(GHBFA)$(Using Table 1, Corollary 4 and Theorem 2 )
\end{proof}
\begin{theorem}
  The picture language family $r_1(L_{d_{HBFA}}(GHBFA))$ equals to  $L_d(GHRFA)$ for $d \in D_2$.
\end{theorem}
\begin{proof}
     $r_1(L_{d_{HBFA}}(GHBFA))$ =  $r_1(L_{d_{HRFA}}(GHRFA))$ (Theorem 4)\\
                                =  $L_{\tiny{\begin{pmatrix}
\CircleArrowleft&s\rotatebox[origin=c]{60}{$\uparrow$}
\end{pmatrix}}}(GHRFA)$(Theorem 6)
=$L_{\tiny{\begin{pmatrix}
 \CircleArrowright&s\rotatebox[origin=c]{60}{$\uparrow$}\\
\end{pmatrix}}}(GHRFA)$(Corollary 10)\\ = $L_{\tiny{\begin{pmatrix}
s  \rotatebox[origin=c]{-60}{$\rightarrow$} &\CircleArrowleft\\
\end{pmatrix}}}(GHRFA)$(Corollary 9) = $L_{\tiny{\begin{pmatrix}
s\rotatebox[origin=c]{60}{$\downarrow$}&\CircleArrowright
    \end{pmatrix}}}(GHRFA)$ (Corollary 10)
\end{proof}
Using Corollary 7
\begin{cor}
 $L_{\tiny{\begin{pmatrix}
s  \rotatebox[origin=c]{60}{$\leftarrow$}&\CircleArrowright \\
\end{pmatrix}}}(GHRFA) = L_{\tiny{\begin{pmatrix}
\CircleArrowleft&s\rotatebox[origin=c]{-60}{$\downarrow$}
\end{pmatrix}}}(GHRFA)$ and\\ $L_{\tiny{\begin{pmatrix}
s\rotatebox[origin=c]{300}{$\uparrow$}&\CircleArrowright
\end{pmatrix}}}(GHRFA)$ =$L_{\tiny{\begin{pmatrix}
 s\rotatebox[origin=c]{300}{$\uparrow$}&\CircleArrowleft\\
\end{pmatrix}}}(GHRFA) $
\end{cor}
Using Corollary 8
\begin{cor}
$L_{\tiny{\begin{pmatrix}
s  \rotatebox[origin=c]{60}{$\leftarrow$}&\CircleArrowright \\
\end{pmatrix}}}(GHRFA) =L_{\tiny{\begin{pmatrix}
 s\rotatebox[origin=c]{300}{$\uparrow$}&\CircleArrowleft\\
\end{pmatrix}}}(GHRFA) $ and\\ $L_{\tiny{\begin{pmatrix}
\CircleArrowleft&s\rotatebox[origin=c]{-60}{$\downarrow$}
\end{pmatrix}}}(GHRFA) =L_{\tiny{\begin{pmatrix}
s\rotatebox[origin=c]{300}{$\uparrow$}&\CircleArrowright
\end{pmatrix}}}(GHRFA) $
\end{cor}
\begin{theorem}
The picture language family $r_5(L_{d_{HBFA}}(GHBFA))$ equals to $L_d(GHBFA)$ for $d \in \Big\{\begin{pmatrix}
 s\rotatebox[origin=c]{300}{$\uparrow$}&\CircleArrowleft\\
\CircleArrowright&\rotatebox[origin=c]{300}{$\downarrow$}
\end{pmatrix}, \begin{pmatrix}
\CircleArrowright &\rotatebox[origin=c]{-60}{$\downarrow$}\\
s\rotatebox[origin=c]{300}{$\uparrow$}&\CircleArrowleft
\end{pmatrix}, \begin{pmatrix}
 \rotatebox[origin=c]{300}{$\uparrow$}&\CircleArrowright\\
\CircleArrowleft&s\rotatebox[origin=c]{-60}{$\downarrow$}
\end{pmatrix}, \begin{pmatrix}
s\rotatebox[origin=c]{60}{$\leftarrow$}&\CircleArrowright \\
\CircleArrowleft & \rotatebox[origin=c]{300}{$\uparrow$}
\end{pmatrix}\Big\}$.
\end{theorem}
\begin{proof}
$r_5(L_{d_{HBFA}}(GHBFA))$ =  $L_{\tiny{\begin{pmatrix}
 s\rotatebox[origin=c]{300}{$\uparrow$}&\CircleArrowleft\\
\CircleArrowright&\rotatebox[origin=c]{300}{$\downarrow$}
\end{pmatrix}}}(GHBFA) $\\
$r_5(L_{d_{HBFA}}(GHBFA))$ =$R_2(r_3(L_{d_{HBFA}}(GHBFA)))$ = $R_2(L_{d_{HBFA}}(GHBFA))$ = $L_{\tiny{\begin{pmatrix}
\CircleArrowright &\rotatebox[origin=c]{-60}{$\downarrow$}\\
s\rotatebox[origin=c]{300}{$\uparrow$}&\CircleArrowleft
\end{pmatrix}}}(GHBFA)$(Using Table 1,Corollary 5 and Theorem 2 )\\
$r_5(L_{d_{HBFA}}(GHBFA))$ =$r_2(R_3(L_{d_{HBFA}}(GHBFA)))$ = $r_2(L_{d_{HBFA}}(GHBFA))$ = $L_{\tiny{\begin{pmatrix}
 \rotatebox[origin=c]{300}{$\uparrow$}&\CircleArrowright\\
\CircleArrowleft&s\rotatebox[origin=c]{-60}{$\downarrow$}
\end{pmatrix}}}(GHBFA)$(Using Table 1,Corollary 6 and Theorem 2 )\\
$r_5(L_{d_{HBFA}}(GHBFA))$ =$R_5(r_0(L_{d_{HBFA}}(GHBFA)))$ = $R_5(L_{d_{HBFA}}(GHBFA))$ = $L_{\tiny{\begin{pmatrix}
s\rotatebox[origin=c]{60}{$\leftarrow$}&\CircleArrowright \\
\CircleArrowleft & \rotatebox[origin=c]{300}{$\uparrow$}
\end{pmatrix}}}(GHBFA)$(Using Table 1,Corollary 4 and Theorem 2 )  
\end{proof}
\begin{theorem}
The picture language family $r_5(L_{d_{HBFA}}(GHBFA))$ is equal to  $L_d(GHRFA)$ for $d \in D_3$.  
\end{theorem}
\begin{proof}
$r_5(L_{d_{HBFA}}(GHBFA))$ =  $r_5(L_{d_{HRFA}}(GHRFA))$ (Theorem 4)\\
                                =  $L_{\tiny{\begin{pmatrix}
 s\rotatebox[origin=c]{300}{$\uparrow$}&\CircleArrowleft\\
\end{pmatrix}}}(GHRFA)$(Theorem 6) =$L_{\scriptstyle{\tiny{\begin{pmatrix}
s\rotatebox[origin=c]{300}{$\uparrow$}&\CircleArrowright
\end{pmatrix}}}}(GHRFA)$(Corollary 11)\\ = $L_{\tiny{\begin{pmatrix}
\CircleArrowleft&s\rotatebox[origin=c]{-60}{$\downarrow$}
\end{pmatrix}}}(GHRFA)$ (Corollary 12)= $L_{\tiny{\begin{pmatrix}
s  \rotatebox[origin=c]{60}{$\leftarrow$}&\CircleArrowright \\
\end{pmatrix}}}(GHRFA)$(Corollary 11)
\end{proof}
\begin{cor}$L_(GHBFA) =L_{\tiny{\begin{pmatrix}
 s\downarrow&\CircleArrowright
\end{pmatrix}}}(GHRFA) \cup L_{\tiny{\begin{pmatrix}
s\rotatebox[origin=c]{60}{$\downarrow$}&\CircleArrowright
    \end{pmatrix}}}(GHRFA) \cup L_{\tiny{\begin{pmatrix}
 s\rotatebox[origin=c]{300}{$\uparrow$}&\CircleArrowleft\\
\end{pmatrix}}}(GHRFA)$.
\end{cor}
 
\begin{cor} $ Op(L(GHBFA)) = L(GHBFA)$, for \\$Op \in H = \{R_1, R_2, R_3, R_4, R_5, r_0, r_1, r_2, r_3,r_4, r_5\}$
\end{cor}
\section{Conclusion
}\label{conclu}

In this paper, we introduced and analyzed two  classes of finite automata operating on hexagonal grids: General Hexagonal Boustrophedon Finite Automata and General Hexagonal Returning Finite Automata. We established formal definitions for both automaton types and characterized the classes of languages they accept.

Our analysis revealed the computational relationships between these classes. The boustrophedon scanning pattern on hexagonal structures provides efficient systematic processing, while the returning capability adds computational flexibility for complex pattern recognition tasks.

The theoretical foundations established here open several avenues for future research. Extensions to probabilistic hexagonal automata, investigation of decidability properties for specific subclasses, and applications to computer vision and image processing represent promising directions. Additionally, exploring the computational complexity of recognition problems for these language classes could yield valuable insights into the practical efficiency of hexagonal grid-based computation.

These results contribute to the broader understanding of two-dimensional automata theory and provide a foundation for developing more sophisticated models that combine geometric advantages of hexagonal connectivity with advanced traversal strategies.

% The documentation for \verb+natbib+ may be found at
% \begin{center}
%   \url{http://mirrors.ctan.org/macros/latex/contrib/natbib/natnotes.pdf}
% \end{center}
% Of note is the command \verb+\citet+, which produces citations
% appropriate for use in inline text.  For example,
% \begin{verbatim}
%    \citet{hasselmo} investigated\dots
% \end{verbatim}
% produces
% \begin{quote}
%   Hasselmo, et al.\ (1995) investigated\dots
% \end{quote}

% \begin{center}
%   \url{https://www.ctan.org/pkg/booktabs}
% \end{center}

% \subsection{Figures}
% \lipsum[10] 
% See Figure \ref{fig:fig1}. Here is how you add footnotes. \footnote{Sample of the first footnote.}
% \lipsum[11] 

% \begin{figure}
%   \centering
%   \fbox{\rule[-.5cm]{4cm}{4cm} \rule[-.5cm]{4cm}{0cm}}
%   \caption{Sample figure caption.}
%   \label{fig:fig1}
% \end{figure}

% \subsection{Tables}
% \lipsum[12]
% See awesome Table~\ref{tab:table}.

% \begin{table}
%  \caption{Sample table title}
%   \centering
%   \begin{tabular}{lll}
%     \toprule
%     \multicolumn{2}{c}{Part}                   \\
%     \cmidrule(r){1-2}
%     Name     & Description     & Size ($\mu$m) \\
%     \midrule
%     Dendrite & Input terminal  & $\sim$100     \\
%     Axon     & Output terminal & $\sim$10      \\
%     Soma     & Cell body       & up to $10^6$  \\
%     \bottomrule
%   \end{tabular}
%   \label{tab:table}
% \end{table}

% \subsection{Lists}
% \begin{itemize}
% \item Lorem ipsum dolor sit amet
% \item consectetur adipiscing elit. 
% \item Aliquam dignissim blandit est, in dictum tortor gravida eget. In ac rutrum magna.
% \end{itemize}

% \section{Conclusion}
% Your conclusion here

% \section*{Acknowledgments}
% This was was supported in part by......

%Bibliography
\bibliographystyle{unsrt}  
\bibliography{references}

\end{document}